\documentclass[letterpaper, 10 pt, conference]{ieeeconf}
\IEEEoverridecommandlockouts
\usepackage{amsmath,amsfonts,amssymb}
\usepackage{graphicx}
\usepackage{cite}
\usepackage{bm}
\usepackage{color}
\usepackage{algorithm}
\usepackage{algorithmic}
\usepackage{booktabs}
\usepackage{multirow}

\newtheorem{theorem}{Theorem}

\newtheorem{assumption}{Assumption}
\newtheorem{remark}{Remark}

\title{\LARGE \bf
Chance-Constrained Nonlinear Covariance Control via Robust Linearization Remainder Bounds
}

\author{Man Jun Koh \and Hyochoong Bang \and SooJean Han%
\thanks{This research was supported by the Space-K BIG Project Program funded by the Korea AeroSpace Administration (KASA) (Grant No. RS-2025-16063273).}%
\thanks{Man Jun Koh and Hyochoong Bang are with the Department of Aerospace Engineering, and SooJean Han is with the School of Electrical Engineering, Korea Advanced Institute of Science and Technology (KAIST), Daejeon, Republic of Korea, 34141. Email: {\tt\small \{mjk18, hcbang, soojean\}@kaist.ac.kr}}%
}

\begin{document}
\maketitle
\thispagestyle{empty}
\pagestyle{empty}

\begin{abstract}
When dealing with nonlinear systems, classical covariance steering typically propagates uncertainty via first-order linearizations, discarding higher-order Taylor remainders. This truncation causes computed statistical moments to diverge from the true physical state distribution, often leading to chance constraint violations. This paper introduces a discrete-time Sequential Convex Programming (SCP) framework that casts the deterministic one-step nonlinear numerical map as a Linear Stochastic Inclusion. The The Taylor remainder is modeled as a state-independent unstructured uncertainty block, bounded over a uniform envelope. The second-moment tubes are propagated via what we refer to as a robust Stochastic Linear Matrix Inequality (S-LMI) derived from the Petersen's lemma, providing an upper bound on the expected uncentered second moment. Domain-exit risk is bounded analytically via a Markov trace inequality, and spatial chance constraints are enforced via Gauss unimodal second-moment bounds within a Difference-of-Convex program. Simulations on a state-dependent nonlinear dynamic system demonstrate constraint satisfaction.
% , whereas Gaussian-propagation baseline fails.
\end{abstract}

\section{Introduction}
Stochastic optimal control of autonomous systems requires satisfying spatial and actuator chance constraints subject to process noise and dynamic nonlinearities. For linear systems with additive Gaussian disturbance, covariance steering provides framework that is tractable by casting the problem as a deterministic optimization over the first two statistical moments\cite{HotzSkelton1987,Chen2016a,Bakolas2018,Okamoto2018}. Recent advances include lossless convexification of the covariance dynamics \cite{liu2024optimal}, computationally efficient formulations that scale linearly with discretization nodes \cite{Kumagai2025,pilipovsky2024cdc} and iterative risk allocation that optimizes per-constraint risk budgets\cite{pilipovsky2021}, and distributionally robust chance constraint certificates \cite{renganathan2023}.
 
Extending covariance steering to nonlinear systems introduces a fundamental difficulty; mapping a probability distribution through nonlinear dynamics couples the mean and covariance, and the resulting moment evolution cannot be computed in closed form. The main approach in the Sequential Convex Programming (SCP) literature linearizes the dynamics about a reference trajectory and propagates covariance via the resulting Jacobian \cite{Ridderhof2019,Lew2020,Benedikter2022,Oguri2024,Kumagai2025}. This first-order approximation discards the higher-order Taylor remainder, implicitly assuming it is negligible. Under state-dependent nonlinearities, the truncation can cause the propagated covariance to diverge from the true physical dispersion, potentially violating chance constraint guarantees.

Several approaches mitigate this linearization error through higher-fidelity moment propagation. The Unscented Transform in greedy \cite{bakolas2020greedy} and learned-dynamics \cite{tsolovikos2021} covariance steering, tube stochastic optimal control with chance-constrained inputs  \cite{ozaki2020}, generalized Polynomial Chaos expansions within SCP \cite{nakka2022}, direct policy optimization over sigma-point trajectories \cite{Howell2021}, nonlinear covariance control via differential dynamic programming \cite{yi2020nonlinear}, continuous-time path distribution optimization via the Girsanov theorem \cite{yu2021}, error bounding through state transition tensors \cite{Qi2025nonlinear}, and non-Gaussian distribution steering via conjugate unscented transformations \cite{Qi2025nongaussian}. These methods improve moment estimation accuracy, but the approximation error between propagated moments and the true distribution is not explicitly bounded. If local nonlinearities become execessive with respect to the modeled expansion order, the estimated moments can potentially underestimate the physical dispersion without warning.
 
Rather than improving the accuracy of moment propagation, we bound the
worst-case effect of the discarded nonlinearity on the propagated moments.
Inspired by the incremental quadratic constraint ($\delta$QC) framework
\cite{Acikmeşe2011,kim2024}, we model the Taylor remainder as an unstructured
norm-bounded uncertainty block and cast the incremental dynamics as a pathwise
linear stochastic inclusion. The contribution of this paper is twofold:
\begin{enumerate}
    \item A discrete-time framework that integrates the nonlinear dynamics via
    a deterministic numerical map and introduces stochasticity strictly at the
    discrete level. This renders the Taylor remainder deterministic, enabling
    its factorization as a norm-bounded uncertainty block with a uniform
    envelope while preserving the martingale difference properties of the
    additive noise.
    \item A robust stochastic LMI (S-LMI), derived via Petersen's lemma
    \cite{Petersen1987}, that propagates a positive-definite decision variable
    upper-bounding the uncentered second moment of the killed
    (domain-truncated) state process for all norm-bounded, state-independent
    remainder realizations, without distributional assumptions on the
    additive noise.
\end{enumerate}
Domain-exit risk is then bounded by applying Markov's inequality to the killed
process, and spatial chance constraints are enforced within a
difference-of-convex (DC) program through Gauss unimodal bounds, which sharpen
the Chebyshev tail bound by a factor of $4/9$ at the cost of a mild
unimodality assumption.
% Chance constraints employ the Gauss inequality for unimodal distributions, providing tighter feasibility regions than Chebyshev bounds while retaining distribution-agnostic validity.

% This framework occupies a distinct position relative to both robust tube MPC and stochastic covariance steering. Unlike robust tube MPC, which enforces geometric containment against bounded deterministic disturbances via Minkowski sums, the S-LMI propagates a probabilistic moment bound that accommodates unbounded stochastic noise. Unlike standard SCP-based covariance steering, which trusts the Jacobian-propagated covariance as exact, the S-LMI explicitly accounts for the remainder-induced moment growth.

% The specific contributions of this paper are:
% \begin{enumerate}
%     \item A discrete-time framework that shifts stochastic injection strictly to the discrete-time sequence, bounding the continuous remainder while preserving martingale difference properties.
%     \item A robust S-LMI that propagates the active uncentered covariance as a decision state, upper-bounding the physical moment without Gaussian assumptions.
%     \item Analytical domain-exit risk bounds formulated via Markov's inequality, decomposing spatial constraints into boundary violation and exit-risk budgets.
%     \item Spatial chance constraints via the Gauss inequality 
% for unimodal distributions, enforced through 
% Difference-of-Convex programs with positive interior margins, 
% reducing conservatism by $2.25\times$ relative to Chebyshev 
% while retaining distribution-agnostic guarantees.
% \end{enumerate}

\section{Problem Formulation}

\subsection{System Dynamics and Discrete Stochasticity}
To prevent continuous-time integration inconsistencies, we start with the deterministic Ordinary Differential Equation (ODE) to model the underlying continuous-time physics:
\begin{equation}
    \dot{x}(t) = f_c(x(t)) + B_c(t) u(t) + E_c \tilde{f}_c(x(t))
    \label{eq:ode}
\end{equation}
where $x(t) \in \mathbb{R}^{n_x}$ is the state vector, $u(t) \in \mathbb{R}^{n_u}$ is the control input, and $E_c \in \mathbb{R}^{n_x \times n_p}$ isolates the unmodeled nonlinear channels $\tilde{f}_c(x(t))$. 

To avoid continuous-time ZOH mapping errors, we partition the horizon into uniform intervals $\Delta t_k$. We adopt a deterministic numerical one-step map $f_d$ (e.g., Runge-Kutta 4) to integrate \eqref{eq:ode}. Then we introduce stochasticity explicitly at the discrete-time level to model unmeasured disturbances integrated over $\Delta t_k$:
\begin{equation}
    x_{k+1} = f_d(x_k, u_k) + w_{d,k}
    \label{eq:discrete_exact}
\end{equation}

\begin{assumption}[Discrete Martingale Difference Noise]
Let $\{\mathcal{F}_k\}_{k \in \mathbb{N}}$ denote the discrete filtration generated by the state and disturbance history up to step $k$. The additive noise $w_{d,k}$ forms a martingale difference sequence adapted to $\mathcal{F}_{k+1}$, satisfying $\mathbb{E}[w_{d,k} \mid \mathcal{F}_k] = 0$ almost surely, with a known conditional covariance bound $\mathbb{E}[w_{d,k} w_{d,k}^T \mid \mathcal{F}_k] \preceq W_k$, where $W_k \succeq 0$\cite{Durrett2019}.
\end{assumption}

\subsection{Chance-Constrained Optimal Control Problem}
The stochastic optimization problem over discrete indices $k \in \{0, \dots, N\}$ requires synthesizing a state-feedback policy $u_k = \pi_k(x_k)$ that minimizes expected costs subject to some kind of probabilistic constraints:
\begin{align}
    \min_{\pi_k(\cdot)} \quad & \mathbb{E} \left[ \ell_f(x_N) + \sum_{k=0}^{N-1} \ell_c(x_k, u_k) \Delta t_k \right] \label{eq:soc_cost} \\
    \text{s.t.} \quad & x_{k+1} = f_d(x_k, u_k) + w_{d,k}, \quad u_k = \pi_k(x_k) \nonumber \\
    & \mathbb{P}\left( h_{x,i}^T x_k \le b_{x,i} \right) \ge 1 - \epsilon_{x,i}, \quad \forall k, i \label{eq:soc_chance_x} \\
    & \mathbb{P}\left( h_{u,j}^T u_k \le b_{u,j} \right) \ge 1 - \epsilon_{u,j}, \quad \forall k, j \label{eq:soc_chance_u} \\
    & \mathbb{E}[x_0] = \mu_0, \quad \mathbb{E}[(x_0-\mu_0)(x_0-\mu_0)^T] \preceq \Sigma_0 \label{eq:soc_init} \\
    & \mathbb{E}[x_N] = \mu_f, \hspace{0.5em} 
    \mathbb{E}[(x_N-\mu_f)(x_N-\mu_f)^T] \preceq \Sigma_f\label{eq:soc_terminal}
\end{align}
where $h_{x,i}, b_{x,i}$ define spatial half-space boundaries, $h_{u,j}, b_{u,j}$ define control limits, and $\epsilon \in (0, 0.5]$ denotes the risk or failure probability budget of violating those constraints.

\section{Incremental Dynamics and Local Envelopes}

\subsection{Separation of Nominal and Stochastic Dynamics}
Global bounding of highly nonlinear system is generally intractable. We therefore decompose the physical process into a deterministic nominal reference and a localized stochastic increments. Let $\hat{z}_k =[\hat{x}_k^T, \hat{u}_k^T]^T$ denote the reference trajectory from the previous iteration of sequential convex programming (SCP), and $\bar{z}_k =[\bar{x}_k^T, \bar{u}_k^T]^T$ denote the current decision variables for the nominal update. The stochastic incremental dynamics are defined as $\eta_k = x_k - \bar{x}_k$ and $\xi_k = u_k - \bar{u}_k$. 

Applying the linear feedback policy $\xi_k = K_k \eta_k$ and a multivariate Taylor expansion of $f_d(x_k, u_k)$ about $\hat{z}_k$ yields the exact mapping:
\begin{align}
    x_{k+1} &= f_d(\hat{z}_k) + J_{x,k} (x_k - \hat{x}_k) + J_{u,k} (u_k - \hat{u}_k) \nonumber \\
    &\quad + E_{k} r_k(x_k, u_k) + w_{d,k}
\end{align}
where $J_{x,k} = \nabla_x f_d|_{(\hat{z}_k)}$, $J_{u,k} = \nabla_u f_d|_{(\hat{z}_k)}$, and $r_k$ is the remainder.

During SCP iterations, we enforce the nominal kinematic update utilizing non-negative virtual control slacks $\nu_k \triangleq v_{def,k}^+ - v_{def,k}^-$ for feasibility\cite{Malyuta2022}:
\begin{equation}
    \bar{x}_{k+1} = f_d(\hat{z}_k) + J_{x,k} (\bar{x}_k - \hat{x}_k) + J_{u,k} (\bar{u}_k - \hat{u}_k) + \nu_k
\end{equation}

\begin{remark}
At exact SCP convergence, the dynamic $L_1$ penalty function ensures $\nu_k \to 0$~\cite{Mao2016}. Our theoretical second-moment certificates apply upon this convergence, ensuring martingale structures are not affected by slack mean-shifts.
\end{remark}

Assuming kinematic convergence ($\nu_k = 0$), the incremental dynamics is written as follows:
\begin{equation}
    \eta_{k+1} = (J_{x,k} + J_{u,k} K_k) \eta_k + E_{k} r_k(x_k, u_k) + w_{d,k}
    \label{eq:discrete_sde}
\end{equation}
And the closed-loop state transition matrix is written as $\Phi_{cl,k}(K_k) = J_{x,k} + J_{u,k} K_k$. 

\subsection{Local Trust-Region Bounding Sets}
We define the spatial validity domain $\mathcal{E}_k(\alpha_{tr})$ using the surrogate uncentered second-moment $\hat{Q}_k$ from the previous SCP. Using a truncation parameter $\alpha_{trunc}$ that defines the Mahalanobis radius $R_{max,k}^2 = n_x / \alpha_{trunc}$, relaxed by $\alpha_{tr} \ge 1$:
\begin{equation}
    \mathcal{E}_k(\alpha_{tr}) \triangleq \left\{ \eta \in \mathbb{R}^{n_x} : \eta^T \hat{Q}_k^{-1} \eta \le R_{max,k}^2 \alpha_{tr} \right\}
\end{equation}
The local validity boundary is parameterized via the regularized invariant geometric matrix $S_k \triangleq R_{max,k}^2 \alpha_{tr} (\hat{Q}_k + \varepsilon_S I)$ with $\varepsilon_S > 0$. We isolate the state deviation via $z_k = (C_{x,k} + C_{u,k} K_k) \eta_k \triangleq C_{cl,k}(K_k) \eta_k$ where $C_{x,k}, C_{u,k}$ are selection matrices.

\begin{assumption}[Unstructured Norm-Bounded Remainder] \label{assump:local_remainder}
For all states residing within $\eta \in \mathcal{E}_k(\alpha_{tr})$, the numerical remainder $r_k$ is factored via an unknown matrix $\Delta_k \in \mathbb{R}^{n_p \times (n_x + n_u)}$\cite{Acikmeşe2011,Kim2025}:
\begin{equation}
    r_k = \Delta_k \hat{\Lambda}_k^{tr} C_{cl,k}(K_k) \eta_k
\end{equation}
where $\Delta_k$ is an unknown deterministic matrix satisfying $\|\Delta_k\|_2 \le 1$, held constant over each step and independent of the realization of $\eta_k$. The fixed diagonal matrix $\hat{\Lambda}_k^{tr} \succ 0$ constitutes a uniform deterministic envelope upper-bounding the local numerical Hessian supremum over $\mathcal{E}_k(\alpha_{tr})$.
\end{assumption}

\begin{remark}[Scope] \label{rem:scope}
State-independence of $\Delta_k$ renders \eqref{eq:discrete_sde} a linear stochastic inclusion with exact moment propagation. The exact Taylor remainder is state-dependent; extending Theorem~\ref{thm:discrete_lmi} to state-correlated realizations requires a modified moment recursion and is left to future work. $\hat{\Lambda}_k^{tr}$ is evaluated at the previous SCP iteration (Algorithm~\ref{alg:scp}, line~4).
\end{remark}

\section{The Robust S-LMI}

To evaluate unconditional expectations of \eqref{eq:discrete_sde} over an unbounded probability space without violating local bounds, we utilize a stopped process. Let the domain-exit stopping time be $\tau \triangleq \inf\{j \ge 1 : \eta_j \notin \mathcal{E}_j(\alpha_{tr})\}$. We isolate interior dispersion via the killed process $\tilde{\eta}_k \triangleq \eta_k \mathbf{1}_{\{\tau > k\}}$. Prior to exit, $\tilde{\eta}_k \tilde{\eta}_k^T = \eta_k \eta_k^T \mathbf{1}_{\{\tau > k\}} \preceq S_k$ holds almost surely.

\begin{theorem}[Robust S-LMI] \label{thm:discrete_lmi}
Assume $\bar{x}_0 = \mu_0$ and $Q_0 = \Sigma_0$. Let the linear synthesis policy be parameterized by the active decision matrix via $L_k = K_k Q_k$. Assuming algorithmic convergence ($\nu_k \to 0$), if there exist symmetric positive-definite decision matrices $Q_{k+1} \succ 0$, $Q_k \succ 0$, variable $L_k$, and a scalar multiplier $m_k > 0$ satisfying the discrete Stochastic Linear Matrix Inequality (S-LMI):
\begin{equation}
\resizebox{\columnwidth}{!}{$
\begin{bmatrix} 
Q_{k+1} - W_k - m_k E_k E_k^T & J_{x,k} Q_k + J_{u,k} L_k & 0 \\ 
* & Q_k & (C_{x,k} Q_k + C_{u,k} L_k)^T \hat{\Lambda}^{tr}_k \\ 
* & * & m_k I_{n_p} 
\end{bmatrix} \succeq 0
$}
\label{eq:slmi_anisotropic}
\end{equation}
alongside the local bounding constraint $Q_k \preceq S_k$, then the sequence $Q_k$ upper-bounds the expected uncentered second moment of the killed process: $\mathbb{E}[\tilde{\eta}_k \tilde{\eta}_k^T] \preceq Q_k$.
\end{theorem}

\begin{proof}
Substituting Assumption \ref{assump:local_remainder} into \eqref{eq:discrete_sde} constructs the uncertain closed-loop transition matrix $A_k(\Delta_k) \triangleq \Phi_{cl,k}(K_k) + E_k \Delta_k \hat{\Lambda}^{tr}_k C_{cl,k}(K_k)$. We seek to enforce the point-wise condition $A_k(\Delta_k) Q_k A_k(\Delta_k)^T + W_k \preceq Q_{k+1}$ across all admissible $\|\Delta_k\|_2 \le 1$.

By applying Petersen's lemma and the S-procedure \cite{Petersen1986,Boyd1994}, this block matrix inequality holds if there exist $m_k > 0$ such that \eqref{eq:slmi_anisotropic} is positive semi-definite. The substitution $L_k = K_k Q_k$ resolves the bilinear dependencies. Because indicator functions obey $\mathbf{1}_{\{\tau > k+1\}} \le \mathbf{1}_{\{\tau > k\}}$, the killed process complies with $\tilde{\eta}_{k+1} \tilde{\eta}_{k+1}^T \preceq \eta_{k+1} \eta_{k+1}^T \mathbf{1}_{\{\tau > k\}}$. Taking the conditional expectation with respect to the filtration $\mathcal{F}_k$:
\begin{align}
    \mathbb{E}[\tilde{\eta}_{k+1} \tilde{\eta}_{k+1}^T \mid \mathcal{F}_k] &\preceq \mathbf{1}_{\{\tau > k\}} \mathbb{E}[ \eta_{k+1} \eta_{k+1}^T \mid \mathcal{F}_k ] \nonumber \\
    &\preceq A_k(\Delta_k) \tilde{\eta}_k \tilde{\eta}_k^T A_k(\Delta_k)^T + W_k
\end{align}
The martingale difference cross-terms vanish. Applying the tower property yields the unconditional bound:
\begin{equation}
    \mathbb{E}[\tilde{\eta}_{k+1} \tilde{\eta}_{k+1}^T] \preceq \mathbb{E}[ A_k(\Delta_k) \mathbb{E}[\tilde{\eta}_k \tilde{\eta}_k^T] A_k(\Delta_k)^T ] + W_k
\end{equation}
Assuming by induction that $\mathbb{E}[\tilde{\eta}_k \tilde{\eta}_k^T] \preceq Q_k$, the deterministic bound guarantees $\mathbb{E}[\tilde{\eta}_{k+1} \tilde{\eta}_{k+1}^T] \preceq Q_{k+1}$.
\end{proof}

\section{Risk Allocation and Chance Constraints}

\subsection{Exit Risk Bound via Markov's Inequality}
To bound the probability of violating $h_{x,i}^T x_k > b_{x,i}$ at step $k$, we decompose via Boole's inequality:
\begin{equation}
    \mathbb{P}(h_{x,i}^T x_k > b_{x,i}) \le \underbrace{\mathbb{P}\big(h_{x,i}^T (\bar{x}_k + \tilde{\eta}_k) > b_{x,i}\big)}_{\text{In-Domain Violation}} + \underbrace{\mathbb{P}(\tau \le k)}_{\text{Exit Risk}}
    \label{eq:chance_decomp}
\end{equation}

The cumulative exit probability is bounded by a union over disjoint exit events. The local domain $\mathcal{E}_j$ is defined by $\eta_j^T \hat{Q}_j^{-1} \eta_j \le R_{max,j}^2 \alpha_{tr}$. Applying the scalar Markov inequality to the trace of the pre-exit process \cite{Durrett2019}:
\begin{equation}
\resizebox{\columnwidth}{!}{$
    \mathbb{P}(\tau \le k) = \sum_{j=1}^k \mathbb{P}(\tau = j) \le \sum_{j=1}^k \frac{\mathbb{E}[\text{Tr}(\hat{Q}_j^{-1} \eta_j \eta_j^T \mathbf{1}_{\{\tau > j-1\}})]}{R_{max,j}^2 \alpha_{tr}} \le \sum_{j=1}^k \frac{\text{Tr}(\hat{Q}_j^{-1} Q_j)}{R_{max,j}^2 \alpha_{tr}}
    $}
    \label{eq:exit_union}
\end{equation}
To ensure the cumulative exit risk respects the budget $\epsilon_{exit}$ over horizon $N$, we enforce a per-step threshold for $k = 1, \ldots, N$:
\begin{equation}
    \text{Tr}(\hat{Q}_k^{-1} Q_k) \le \frac{\epsilon_{exit}}{N} \cdot R_{max,k}^2 \alpha_{tr}
    \label{eq:markov_exit}
\end{equation}
The uniform allocation $\epsilon_{exit}/N$ is conservative relative to an optimized per-step allocation, but eliminates $N$ additional decision variables and preserves the convex subproblem structure. The remaining spatial and actuator budgets are:
\begin{equation}
    \epsilon_{c,x,i} = \epsilon_{x,i} - \epsilon_{exit} > 0, \quad \epsilon_{c,u,j} = \epsilon_{u,j} - \epsilon_{exit} > 0
    \label{eq:risk_budgets}
\end{equation}

\subsection{Spatial and Control Chance Constraints}

\begin{assumption}[Unimodal Spatial Dispersion]
\label{asm:unimodal}
Over short discrete intervals $\Delta t_k$, continuous physical systems subject to bounded nonlinearities and additive unimodal process noise preserve a continuous, unimodal state distribution. We assume the projection $h_{x,i}^T \tilde{\eta}_k$ is continuous and unimodal with mode at zero. The stopping time $\tau$ asymmetrically truncates the tails, but for small exit-risk budgets ($\epsilon_{exit} \ll 1$), the truncated mass is $O(\epsilon_{exit})$ and the mode perturbation is negligible. If unimodality is violated, the framework defaults to a Chebyshev bound, recovering standard distribution-agnostic guarantees.
\end{assumption}

Under Assumption~\ref{asm:unimodal}, we evaluate the spatial margin $m_{i,k}^x \triangleq b_{x,i} - h_{x,i}^T \bar{x}_k > 0$ using the Gauss inequality for unimodal distributions \cite{Dharmadhikari1988}. Because the mode resides at zero, the two-sided Gauss bound depends on the uncentered second moment, circumventing the mean-shift induced by asymmetric killing. Given $Q_k \succeq \mathbb{E}[\tilde{\eta}_k \tilde{\eta}_k^T]$:
\begin{equation}
\resizebox{\columnwidth}{!}{$
    \mathbb{P}(h_{x,i}^T \tilde{\eta}_k > m_{i,k}^x) \le \mathbb{P}(|h_{x,i}^T \tilde{\eta}_k| \ge m_{i,k}^x) \le \frac{4}{9} \frac{\mathbb{E}[(h_{x,i}^T \tilde{\eta}_k)^2]}{(m_{i,k}^x)^2} \le \frac{4}{9} \frac{h_{x,i}^T Q_k h_{x,i}}{(m_{i,k}^x)^2}
    $}
    \label{eq:gauss_bound}
\end{equation}
The Gauss bound \eqref{eq:gauss_bound} requires $(m_{i,k}^x)^2 \ge \tfrac{4}{3}\, h_{x,i}^T Q_k h_{x,i}$. This condition is implied by \eqref{eq:mean_free_bound} whenever $\epsilon_{c,x,i} \le 1/3$, which holds for any practical risk budget.

Restricting \eqref{eq:gauss_bound} below the risk threshold $\epsilon_{c,x,i}$ yields:
\begin{equation}
    h_{x,i}^T Q_k h_{x,i} \le \kappa_{g} \cdot \epsilon_{c,x,i} (b_{x,i} - h_{x,i}^T \bar{x}_k)^2 \label{eq:mean_free_bound}
\end{equation}
where $\kappa_{g} \triangleq 9/4 = 2.25$. Compared to the standard Chebyshev bound ($\kappa_g = 1$), unimodality permits $2.25\times$ larger projected second moment before the constraint activates.

\begin{remark}[Connection to Robust Tube MPC and DRO]
\label{rem:tube_dro}
Robust tube MPC enforces geometric containment via Minkowski sums against bounded deterministic disturbances. Our S-LMI framework instead propagates a probabilistic moment bound, accommodating unbounded stochastic noise. Constraint \eqref{eq:mean_free_bound} evaluates risk analogously to moment-based DRO \cite{Calafiore2006},but restricting the ambiguity set to unimodal distributions via $\kappa_g$ reduces the artificial trajectory compression typical of pure robust frameworks while retaining distribution-agnostic validity.
% This expands the feasible variance by $2.25\times$ relative to Chebyshev, reducing the artificial trajectory compression typical of pure robust frameworks while retaining distribution-agnostic validity.
\end{remark}

We linearize the concave term $(m_{i,k}^x)^2$ around the previous reference $\hat{m}_{i,k}^x$ via the Convex-Concave Procedure (CCP) with slack $s_{i,k}^x \ge 0$ \cite{Yuille2003}:
\begin{equation}
    h_{x,i}^T Q_k h_{x,i} \le \kappa_{g} \cdot \epsilon_{c,x,i} \Big[ (\hat{m}_{i,k}^x)^2 + 2 \hat{m}_{i,k}^x (m_{i,k}^x - \hat{m}_{i,k}^x) \Big] + s_{i,k}^x \label{eq:ccp_var_x}
\end{equation}

\begin{remark}[Margin Positivity]
\label{rem:margin}
The probability bound \eqref{eq:gauss_bound} requires $m_{i,k}^x > 0$. If $m_{i,k}^x \le 0$, the squared margin $(m_{i,k}^x)^2$ creates a false feasible region. The CCP linearization \eqref{eq:ccp_var_x} does not suffer this artifact---it evaluates negatively when the margin drops below zero while the reference remains positive---but strict positivity is enforced to preserve the validity of \eqref{eq:gauss_bound}:
\end{remark}
\begin{equation}
    b_{x,i} - h_{x,i}^T \bar{x}_k \ge \varepsilon_{margin} \label{eq:strict_margin_x}
\end{equation}

For active feedback $u_k = \bar{u}_k + K_k \tilde{\eta}_k$, bounding control energy against the worst-case ellipsoid $S_k$ is overly conservative. Instead, we bound against the active moment $Q_k$ via an auxiliary matrix $U_k$ and the Schur complement \cite{Boyd1994}:
\begin{equation}
    \begin{bmatrix} U_k & L_k \\ L_k^T & Q_k \end{bmatrix} \succeq 0 \quad \implies \quad K_k Q_k K_k^T \preceq U_k
    \label{eq:input_schur}
\end{equation}
Since linear projections preserve unimodality, the Gauss bound applies to the actuator channels with the same multiplier $\kappa_g$. Actuator limits enforce DC constraints with margins $m_{j,k}^u \ge \varepsilon_{margin}$ and slacks $s_{j,k}^u \ge 0$:
\begin{equation}
\resizebox{\columnwidth}{!}{$
    h_{u,j}^T U_k h_{u,j} \le \kappa_{g} \cdot \epsilon_{c,u,j} \Big[ (\hat{m}_{j,k}^u)^2 + 2 \hat{m}_{j,k}^u (m_{j,k}^u - \hat{m}_{j,k}^u) \Big] + s_{j,k}^u
$}
    \label{eq:ccp_var_u}
\end{equation}

\section{Sequential Convex Programming}

The S-LMI tube is computed iteratively via Successive Convexification (SCvx) \cite{Mao2016}. At each iteration, the nonlinear map $f_d$ is linearized about the current reference $(\hat{x}_k, \hat{u}_k)$, the envelope $\hat{\Lambda}_k^{\mathrm{tr}}$ is evaluated over $\mathcal{E}_k(\alpha_{\mathrm{tr}})$, and the CCP margins are updated from the previous nominal trajectory. The surrogate subproblem optimizes the active variables $\Theta = \{ \bar{x}_k, \bar{u}_k, Q_k, U_k, L_k, m_k \}$ and non-negative $L_1$ penalty slacks $\mathcal{S} = \{ v_{def}^{\pm}, s^x, s^u, s^{gnd}, v_{mkv}, E_{env} \}$:
\begin{align}
    \min_{\Theta, \mathcal{S}} \quad & \sum_{k=0}^{N-1} \Big( \omega_{\ell} \|\bar{u}_k\|_2^2 + \omega_s \|\bar{u}_k - \bar{u}_{k-1}\|_2^2 \nonumber \\
    &+ \omega_{prox} \|\bar{u}_k - \hat{u}_k\|_2^2 + \omega_u \text{Tr}(U_k) \Big) + \omega_d \sum \| \mathcal{S} \|_1 \label{eq:scvx_cost} \\
    \text{s.t.} \quad & \bar{x}_{k+1} = f_d(\hat{z}_k) + J_{x,k} (\bar{x}_k - \hat{x}_k) + J_{u,k} (\bar{u}_k - \hat{u}_k) \nonumber \\
    &\quad + v_{def,k}^+ - v_{def,k}^- \nonumber \\
    & \text{Constraints } \eqref{eq:slmi_anisotropic}, \eqref{eq:markov_exit}, \eqref{eq:ccp_var_x}, \eqref{eq:strict_margin_x}, \eqref{eq:input_schur}, \eqref{eq:ccp_var_u} \nonumber \\
    & R_{max}^2 \alpha_{tr} (\hat{Q}_k + \varepsilon_S I) - Q_k + E_{env,k} \succeq 0 \nonumber \\
    & \bar{x}_0 = \mu_0, \quad Q_0 = \Sigma_0, \quad \bar{x}_N = \mu_f \nonumber \\
    & \|\bar{x}_k - \hat{x}_k\|_{\infty} \le \Delta_{tr}^x, \quad \|\bar{u}_k - \hat{u}_k\|_{\infty} \le \Delta_{tr}^u \nonumber \\
    & \mathcal{S} \ge 0, \ Q_k \succ 0, \ U_k \succ 0, \ m_k > 0 \nonumber
\end{align}
The cost penalizes control effort ($\omega_\ell$), control smoothness ($\omega_s$), proximal deviation ($\omega_{prox}$), and feedback energy $\text{Tr}(U_k)$ ($\omega_u$), while the trace of $Q_k$ is not penalized to prevent artificial dispersion compression. The dynamic penalty weight $\omega_d$ is scaled by $\beta_\omega > 1$ after each iteration to drive the slacks toward zero.

Trajectory viability is evaluated via the acceptance ratio $\rho = \Delta J_{\mathrm{act}} / \Delta J_{\mathrm{pred}}$, where $\Delta J_{\mathrm{act}}$ and $\Delta J_{\mathrm{pred}}$ are the actual and predicted cost reductions between successive iterates. If $\rho < \rho_{\min}$, the update is rejected and trust regions are contracted by $\gamma_c$; if $\rho > \rho_{\max}$, trust regions are expanded by $\gamma_e$. Convergence is declared when the control change, kinematics defect, and chance constraint slacks fall below prescribed tolerances.

\begin{algorithm}[H]
\caption{S-LMI Synthesis}
\label{alg:scp}
\begin{algorithmic}[1]
\STATE \textbf{Initialize:} $\hat{x}_k, \hat{u}_k$ via deterministic trajectory optimization. $K_k = 0$, $\hat{Q}_k$ via open-loop Jacobian propagation. Set $\omega_d$, $\Delta_{tr}^x$, $\Delta_{tr}^u$.
\FOR{$i = 1$ \TO $i_{\max}$}
    \STATE Compute Jacobians $J_{x,k}, J_{u,k}$ at $(\hat{x}_k, \hat{u}_k)$.
    \STATE Compute $\hat{\Lambda}_k^{\mathrm{tr}}$ over $\mathcal{E}_k(\alpha_{\mathrm{tr}})$.
    \STATE Update CCP linearization points from $\hat{x}_k$.
    \STATE Solve surrogate SDP \eqref{eq:scvx_cost} $\rightarrow \bar{x}_k, \bar{u}_k, Q_k, L_k$.
    \STATE Evaluate $\rho = \Delta J_{\mathrm{act}} / \Delta J_{\mathrm{pred}}$.
    \IF{$\rho < \rho_{\min}$}
        \STATE Reject. $\Delta_{tr} \leftarrow \gamma_c\, \Delta_{tr}$, $\omega_d \leftarrow \beta_\omega\, \omega_d$.
    \ELSE
        \STATE Accept. $\hat{x}_k \leftarrow \bar{x}_k$, $\hat{u}_k \leftarrow \bar{u}_k$, $\hat{Q}_k \leftarrow Q_k$.
        \STATE If $\rho > \rho_{\max}$: $\Delta_{tr} \leftarrow \gamma_e\, \Delta_{tr}$; else contract.
        \IF{convergence tolerances met}
            \STATE \textbf{break}
        \ENDIF
    \ENDIF
\ENDFOR
\RETURN $K_k^{\star} = L_k Q_k^{-1}$, $\bar{u}_k^{\star}$, $Q_k^{\star}$.
\end{algorithmic}
\end{algorithm}

\section{Numerical Simulation}\label{sec:sim}
% \vspace{-0.031cm}
We evaluate the proposed S-LMI framework on an EDL-inspired double-integrator problem with gravitational acceleration, quadratic aerodynamic drag, and bilinear position-velocity coupling. We compare against the iterative covariance steering (iCS) method~\cite{Ridderhof2019}, which uses first-order Jacobian linearization with Gaussian covariance propagation and chance constraints. Both methods share the same SCP architecture and affine feedback policy; comparing against iCS isolates the effect of the Taylor remainder bounds. While the closest alternative baseline is the gPC-SCP approach \cite{nakka2022}, it differs with its open-loop synthesis and polynomial chaos propagation which would conflate the isolated effect of our remainder bounds with broader architectural differences. 

We used an LG Gram with an Intel i7-1260P and 16 GB RAM, and both trajectories were generated using Clarabel \cite{Clarabel_2024} and CVXPY \cite{cvxpy}.

% We evaluate the proposed S-LMI framework on an EDL-inspired double-integrator problem with quadratic aerodynamic drag, and bilinear position-velocity coupling. We compare against the iterative covariance steering (iCS) method~\cite{Ridderhof2019}. We used an LG Gram with Intel i7-1260P and 16 GB RAM, and both trajectories were generated using Clarabel \cite{Clarabel_2024} and CVXPY \cite{cvxpy}.

% \begin{remark}[Baseline selection]
%     Both iCS and the proposed S-LMI share the same SCP architecture, first-order Jacobian linearization, and affine state-feedback synthesis; the sole difference is the treatment of the Taylor remainder. The closest alternative baseline, gPC-SCP \cite{nakka2022}, differs structurally by synthesizing open-loop controls and propagating uncertainty via polynomial chaos expansion. Comparing against gPC-SCP would conflate the benefits of our remainder bounding mechanism with broader differences in control architecture.
% \end{remark}
% comparing against iCS isolates the effect of the Taylor remainder bounds without introducing confounding variables from fundamentally different uncertainty propagation schemes (e.g., unscented transforms, generalized polynomial chaos).

\subsection{System Dynamics}

The state $x = [\xi_1, \xi_2, v_1, v_2]^T \in \mathbb{R}^4$ consists of lateral position~$\xi_1$, altitude~$\xi_2$, and their velocities. The thrust input is $u = [u_1, u_2]^T \in \mathbb{R}^2$. The continuous-time dynamics are
\begin{equation}\label{eq:edl_dyn}
\dot{x} = \begin{bmatrix} v_1 \\ v_2 \\ u_1 - c_d \|v\| v_1 + \alpha_1 \xi_1 v_1 \\ u_2 - g - c_d \|v\| v_2 + \alpha_2 \xi_2 v_2 \end{bmatrix}
\end{equation}
where $g = 1.0$, $c_d = 0.005$, and $\alpha_1 = 0.03$, $\alpha_2 = 0.01$ are bilinear coupling strengths. Additive process noise with intensity $\gamma = 0.015$ enters the velocity channels, with $W_k$ computed via Van~Loan. The map $f_d$ uses RK4 with 10 sub-steps.

\subsection{Problem Setup}

The planning horizon is $t_f = 12$\,s discretized into $N = 25$ steps. Boundary conditions are $\bar{x}_0 = [1.0,\, 15.0,\, 2.3,\, {-}1.0]^T$, $\bar{x}_f = [1.0,\, 0.0,\, 0.0,\, 0.0]^T$, with initial covariance $\Sigma_0 = 0.025\, I_4$ and terminal covariance bound $\Sigma_f = 0.05\, I_4$. Control inputs are bounded by $|u_j| \leq 2.0$. Two spatial chance constraints define a descent corridor: lateral walls $\Pr(|\xi_1| \leq 3.8) \geq 0.95$ and ground $\Pr(\xi_2 \geq {-}0.2) \geq 0.95$ at each interior step ($k = 1,\dots,N{-}1$). The risk budget is $\epsilon_{\mathrm{exit}} = 0.01$, $\epsilon_c = 0.04$.

% Both are enforced via the CCP procedure with the Gauss unimodal multiplier $\kappa_g = 9/4$ and margin $\varepsilon_{\mathrm{margin}} = 0.2$. The total risk budget $\epsilon_{\mathrm{total}} = 0.05$ is split into exit risk $\epsilon_{\mathrm{exit}} = 0.01$ and spatial budget $\epsilon_c = 0.04$.

% The uncertainty extraction matrix is $E_c E_c^T = \mathrm{diag}(10^{-4},\, 10^{-4},\, 1,\, 1)$, reflecting that the nonlinear remainder affects velocity channels only (position kinematics $\dot{\xi} = v$ and gravity are linear). The Markov per-step exit bound requires $R^2_{\max}\alpha_{\mathrm{tr}} \geq n_x N / \epsilon_{\mathrm{exit}} = 10{,}000$. For the dynamics~\eqref{eq:edl_dyn}, $\hat{\Lambda}_k^{\mathrm{tr}}$ is computed analytically. The Hessian of each bilinear coupling $\alpha_i \xi_i v_i$ is exactly $\alpha_i$ (constant, state-independent). The Hessian of the quadratic drag $c_d \|v\| v_i$ is bounded by $3 c_d = 0.015$. Gravity contributes zero to the Hessian (constant term). The per-channel continuous-time bounds are $\alpha_1 + 3c_d = 0.045$ (channel~1) and $\alpha_2 + 3c_d = 0.025$ (channel~2). The discrete-time bound follows from Gronwall amplification. This closed-form computation requires no sampling or safety factors. Algorithm~1 is initialized with $K_k = 0$.
The uncertainty extraction matrix is $E_c E_c^T = \mathrm{diag}(10^{-4},\, 10^{-4},\, 1,\, 1)$, and $R^2_{\max}\alpha_{\mathrm{tr}} = 10{,}000$. The envelope $\hat{\Lambda}_k^{\mathrm{tr}}$ is computed analytically: the Hessian of $\alpha_i \xi_i v_i$ is exactly $\alpha_i$ (constant), the drag Hessian is bounded by $3 c_d = 0.015$, and gravity contributes zero. Per-channel bounds are $\alpha_1 + 3c_d = 0.045$ and $\alpha_2 + 3c_d = 0.025$, amplified via Gronwall for the discrete-time map. Algorithm~\ref{alg:scp} is initialized with $K_k = 0$, and uses SCvx parameters $\rho_{\min} = 0.05$, $\rho_{\max} = 0.7$, $\gamma_c = 0.5$, $\gamma_e = 1.2$, $\beta_\omega = 1.2$.

% \paragraph{Analytical envelope $\hat{\Lambda}_k^{\mathrm{tr}}$.}

\subsection{Results}

Table~\ref{tab:metrics} summarizes the $5{,}000$ Monte carlo runs. S-LMI converges in $7$ SCP iterations ($215$\,s) without any rejects; iCS converges in $6$ iterations ($719$\,s). Fig.~\ref{fig:traj_comparison} shows the position-space trajectories. Under iCS (Fig.~\ref{fig:traj_comparison}a), trajectories breach both the lateral wall and the ground constraint. The Gaussian covariance ellipses remain small because the quantile constraint compresses the predicted variance, but the actual dispersion exceeds these predictions near the ground where the bilinear dynamics produce non-Gaussian spreading. Under S-LMI (Fig.~\ref{fig:traj_comparison}b), all trajectories remain within both constraint boundaries, and the $Q_k$ bounds contain the Monte Carlo scatter. Fig.~\ref{fig:states_controls} shows that both controllers produce similar nominal trajectories and control profiles, with comparable thrust authority within the $|u_j| \leq 2$ bound. The S-LMI $Q_k$ envelope (shaded) envelops the closed-loop spread in all channels.

\begin{figure}[t]
  \centering
  \includegraphics[width=\linewidth]{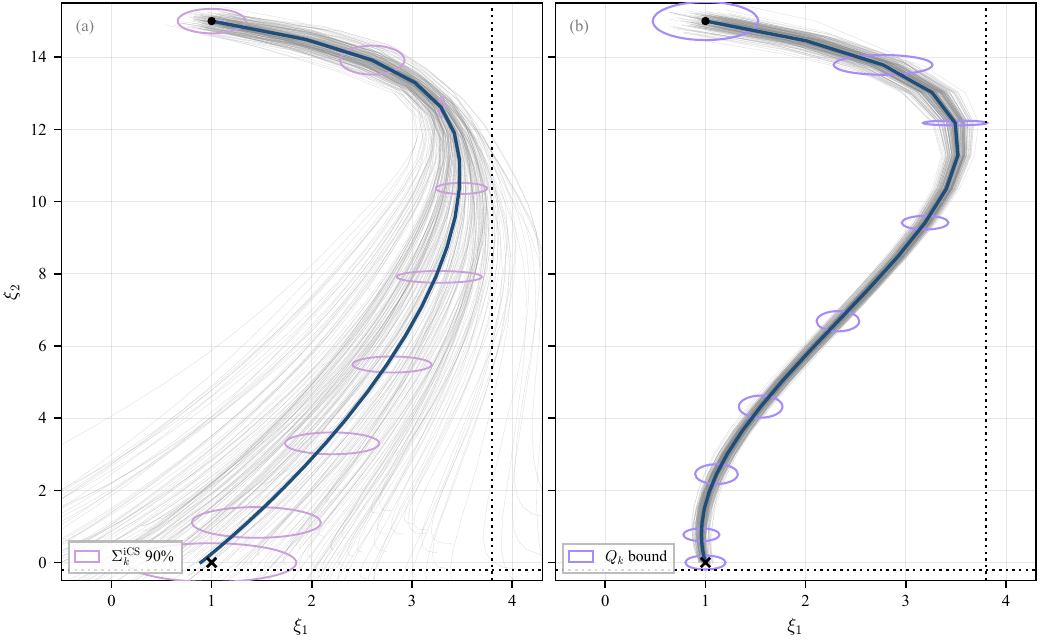}
  \caption{Monte Carlo trajectories ($5{,}000$ runs). (a)~iCS with $\Sigma_k^{\mathrm{iCS}}$ ellipses ($90\%$). (b)~S-LMI with $Q_k$ bounds. Dotted lines: $|\xi_1| \leq 3.8$, $\xi_2 \geq -0.2$.}
  \label{fig:traj_comparison}
\end{figure}

\begin{figure}[t]
  \centering
  \includegraphics[width=\linewidth]{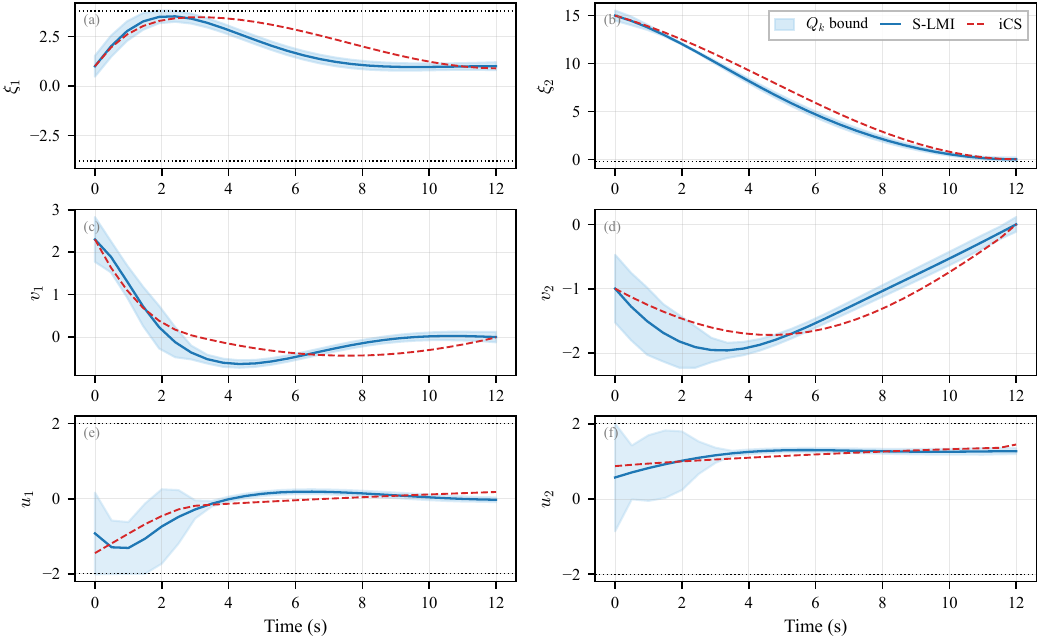}
  \caption{State (a--d) and control (e--f) profiles. Blue: S-LMI; red dashed: iCS; shaded: $Q_k$ bound.}
  \label{fig:states_controls}
\end{figure}

Fig.~\ref{fig:viol_cov}a shows the per-step violation rate. The iCS empirical violation reaches $40.3\%$ at the terminal step while its Gaussian prediction peaks at $21.1\%$; the factor-of-two discrepancy reflects the Gaussian tail model's limitation under the bilinear-gravitational dynamics. S-LMI maintains violations at $0.04\%$. Fig.~\ref{fig:viol_cov}b confirms that $\mathrm{Tr}(Q_k)$ upper-bounds the empirical Monte Carlo trace at every step with a tightness ratio of $1.01\times$.

\begin{figure}[t]
  \centering
  \includegraphics[width=\linewidth]{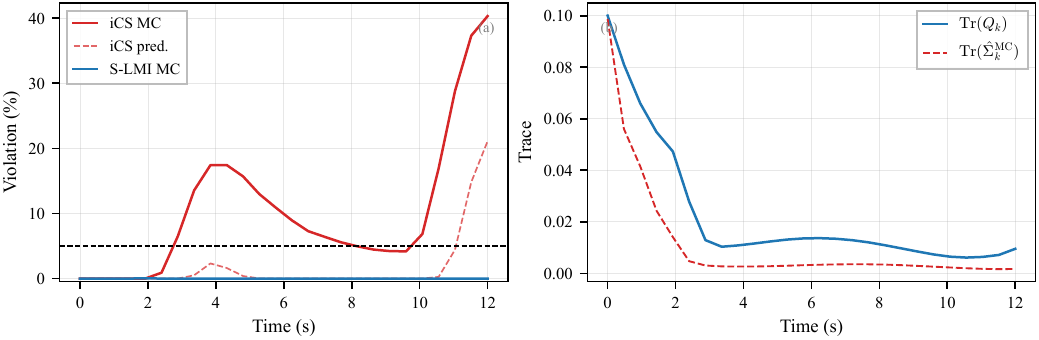}
  \caption{(a)~Per-step violation rate; dashed black: $5\%$ target. (b)~$\mathrm{Tr}(Q_k)$ vs.\ empirical MC trace (S-LMI only).}
  \label{fig:viol_cov}
\end{figure}

% Fig.~\ref{fig:effort_gains} compares per-step control effort and feedback gains. Both methods follow similar effort profiles with comparable totals ($41.5$ vs.\ $40.0$), confirming that S-LMI achieves constraint satisfaction through the remainder bounding mechanism rather than additional control authority. S-LMI uses higher gains during the initial descent to establish the covariance bound, then settles to levels comparable to iCS; peak gains are similar ($3.17$ vs.\ $3.40$).
Fig.~\ref{fig:effort_gains} show comparable control effort (41.5 vs. 40.0) and peak feedback gains (3.17 vs. 3.40); this shows that constraint satisfaction comes from the remainder bounding mechanism rather than additional control authority.

\begin{figure}[t]
  \centering
  \includegraphics[width=\linewidth]{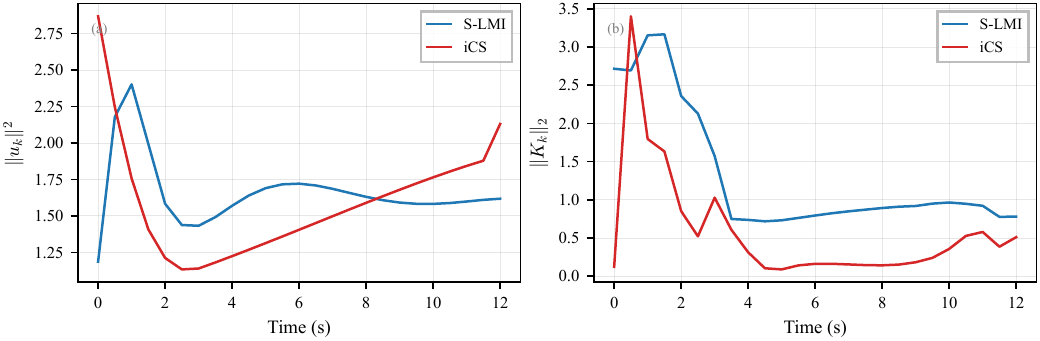}
  \caption{(a)~Per-step control effort $\|u_k\|^2$. (b)~Feedback gain $\|K_k\|_2$.}
  \label{fig:effort_gains}
\end{figure}

\begin{table}[t]
  \centering
  \caption{S-LMI vs.\ iCS ($5{,}000$ Monte Carlo trials).}
  \label{tab:metrics}
  \begin{tabular}{lcc}
    \hline
    Metric & iCS & S-LMI \\
    \hline
    Max MC violation (\%) & 40.3 & 0.04 \\
    Control effort $\sum\|u_k\|^2$ & 40.0 & 41.5 \\
    Max $\|K_k\|_2$ & 3.40 & 3.17 \\
    Total run time (s) & 719 & 215 \\
    Iterations & 6 & 7 \\
    Max-trace ratio & -- & $1.01\times$ \\
    Terminal mean error & 0.050 & 0.001 \\
    \hline
  \end{tabular}
\end{table}

\section{Conclusion}
This paper presented an SCP framework for chance-constrained nonlinear covariance control that bounds the Taylor remainder via Petersen's lemma, propagating an upper bound on the uncentered second moment without distributional assumptions. On a corridor descent problem with bilinear coupling and gravity, the S-LMI achieves $0.04\%$ empirical violation with $1.01\times$ bound tightness and comparable control effort to the iCS baseline, completing in $7$ SCP iterations.

The envelope $\hat{\Lambda}_k^{\mathrm{tr}}$ was computed analytically, general dynamics require systematic computation via interval arithmetic and automatic differentiation. Future work includes optimizing per-step risk allocation, exploiting Hessian sparsity through structured $\delta QC$ multipliers, and validation on entry or powered descent guidance problems.

% This paper presented a discrete-time sampled-data SCP framework for chance-constrained optimal control of nonlinear dynamic systems. By isolating the continuous remainder of the deterministic numerical one-step map and modeling it as an unstructured uncertainty block, the proposed methodology constructs a causal S-LMI that upper-bounds the expected uncentered second moment of the state. Analytical domain-exit risk bounding via Markov's inequality and mean-free DC constraints ensure chance constraint satisfaction without relying on certainty-equivalence Gaussian assumptions. Numerical simulations on a cross-coupled nonlinear trajectory optimization problem confirm that the framework achieves strict constraint compliance, outperforming standard Gaussian propagation methods.

\bibliographystyle{IEEEtran}
\bibliography{references.bib}

\end{document}